\documentclass[a4paper,twoside,11pt,reqno]{amsart}
\usepackage{fullpage}
\usepackage[T1]{fontenc}
\usepackage[utf8]{inputenc}
\usepackage{enumitem}

\usepackage{hyperref}
\usepackage[baseline]{euflag}
\usepackage[slantedGreeks, partialup, noDcommand]{kpfonts}
\usepackage{graphicx}
\usepackage[mathscr]{euscript}
\usepackage{caption}
\usepackage{tikz}
\usepackage{comment}
\usetikzlibrary{trees}

\hypersetup{ocgcolorlinks=true,allcolors=testc}
\hypersetup{
     colorlinks   = true,
     citecolor    = black
}
\hypersetup{linkcolor=black}
\hypersetup{urlcolor=black}

\newcommand{\eqdef}{\stackrel{\mathrm{def}}{=}}

\newtheorem{theorem}{Theorem}
\newtheorem{lemma}[theorem]{Lemma}
\newtheorem{question}[theorem]{Question}
\newtheorem{proposition}[theorem]{Proposition}
\newtheorem{corollary}[theorem]{Corollary}

\newtheorem{remark}[theorem]{Remark}
\newtheorem{definition}[theorem]{Definition}

\newtheorem*{definition*}{Definition}

\newcommand*\diff{\mathop{}\!\mathrm{d}}
\newcommand{\R}{\mathbb{R}} 

\newcommand{\N}{\mathbb{N}}

\newcommand{\e}{\varepsilon}

\usepackage{dutchcal}

\newcommand{\E}{\mathbb{E}}
\newcommand{\Tsf}{\mathsf{T}}

\newcommand{\M}{\mathsf{M}}

\newcommand{\op}[1] {
\big\| #1\big\|_{\mathrm{op}}
}

\newcommand{\p}[1]{
{#1}^{\prime}
}

\newcommand{\Var}[1] {
\mathrm{Var}#1
}

\newcommand{\Ee} {
E_{\varepsilon}
}
\newcommand{\ix}[1] {
\langle \mathcal{I}(#1)x,x\rangle
}

\newcommand{\mb}{\mathbb}
\newcommand{\ms}{\mathscr}
\newcommand{\msf}{\mathsf}
\newcommand{\mr}{\mathrm}
\newcommand{\mc}{\mathcal}

\title{On the entropy and information of Gaussian mixtures}

\author{Alexandros Eskenazis}
\address{(A.~E.) CNRS, Institut de Math\'ematiques de Jussieu, Sorbonne Universit\'e, France and Trinity College, University of Cambridge, UK.}
\email{alexandros.eskenazis@imj-prg.fr, ae466@cam.ac.uk}

\author{Lampros Gavalakis}
\address{(L.~G.) Univ Gustave Eiffel, Univ Paris Est Creteil, CNRS, LAMA UMR8050 F-77447 Marne-la-Vall{\'e}e, France. }
\email{lampros.gavalakis@univ-eiffel.fr} 
\thanks{L.G. has received funding from the European Union's Horizon 2020 research and innovation program
 under the Marie Sklodowska-Curie grant agreement No 101034255~{\large \euflag} and by the B{\'e}zout Labex, funded by ANR, reference ANR-10-LABX-58.}

\begin{document}

\maketitle

\begin{abstract}
We establish several convexity properties for the entropy and Fisher information of mixtures of centered Gaussian distributions. First, we prove that if $X_1, X_2$ are independent scalar Gaussian mixtures, then the entropy of $\sqrt{t}X_1 + \sqrt{1-t}X_2$  is concave in $t \in [0,1]$, thus confirming a conjecture of Ball, Nayar and Tkocz (2016) for this class of random variables. In fact, we prove a generalisation of this assertion which also strengthens a result of Eskenazis, Nayar and Tkocz (2018). For the Fisher information, we extend a convexity result of Bobkov (2022) by showing that the Fisher information matrix is operator convex as a matrix-valued function acting on densities of mixtures in $\mathbb{R}^d$. As an application, we establish rates for the convergence of the Fisher information matrix of the sum of weighted i.i.d.~Gaussian mixtures in the operator norm along the central limit theorem under mild moment assumptions.
\end{abstract}

\bigskip

{\footnotesize

}


\section{Introduction}

\subsection{Entropy}
Let $X$ be a continuous random vector in $\mathbb{R}^d$ with density $f:\R^d\to\R_+$. The (differential) entropy of $X$ is the quantity 
\begin{equation} 
h(X) \eqdef -\int_{\mathbb{R}^d} f(x)\log{f(x)}\, \diff x = \mb{E} \big[ - \log f(X)\big],
\end{equation}
where $\log$ always denotes the natural logarithm.
The celebrated entropy power inequality of Shannon and Stam \cite{Sha48,Sta59} (see also \cite{Lie78}) implies that for every independent continuous random vectors $X_1,X_2$ in $\R^d$, we have
\begin{equation} \label{eq:epi}
\forall \ t\in[0,1],\qquad h\big(\sqrt{t}X_1 + \sqrt{1-t}X_2\big) \geq th(X_1) + (1-t)h(X_2).
\end{equation}
In general, the entropy power inequality cannot be reversed (see, e.g., the construction of \cite[Proposition 4]{BC15}). However, reverse entropy power inequalities have been considered under different assumptions on the random vectors, such as log-concavity \cite{bobkovreverse, bobkov2013problem, ball2016reverse, madimanforward}.

It follows directly from \eqref{eq:epi} that if $X_1, X_2$ are i.i.d.~random vectors, then the entropy function $t\mapsto  h(\sqrt{t}X_1 + \sqrt{1-t}X_2)$ is minimised at $t=0$ and $t=1$. In the spirit of reversing the entropy power inequality, Ball, Nayar and Tkocz \cite{ball2016reverse} raised the question of maximising this function. In particular, they gave an example of a random variable $X_1$ for which the maximum is not attained at $t=\frac{1}{2}$ but conjectured that for i.i.d.~log-concave random variables this function must be concave in $t\in[0,1]$, in which case it is in particular maximised at $t=\frac{1}{2}$. It is worth noting that the conjectured concavity would also be a \emph{strengthening} of the entropy power inequality for i.i.d.~random variables, as \eqref{eq:epi} amounts to the concavity condition for the points $0, t, 1$. So far, no special case of the conjecture of \cite{ball2016reverse} seems to be known.

In this work, we consider (centered) Gaussian mixtures, i.e.~random variables of the form
\begin{equation}
X = YZ,
\end{equation}
where $Y$ is an almost surely positive random variable and $Z$ is a standard Gaussian random variable, independent of $Y$. The resulting random variable has density of the form 
\begin{equation} \label{eq:dens-mixt}
\forall \ x\in\R, \qquad f_X(x) = \E\big[{\frac{1}{\sqrt{2\pi Y^2}}e^{-\frac{x^2}{2Y^2}}}\big].
\end{equation}
In particular, as observed in \cite{eskeGM}, \eqref{eq:dens-mixt} combined with Bernstein's theorem readily implies that a symmetric random variable $X$ is a Gaussian mixture if and only if $x\mapsto f_X(\sqrt{x})$ is completely monotonic on $(0,\infty)$. Therefore, distributions with density proportional to $e^{-|x|^p}$, symmetric $p$-stable random variables, where $p\in(0,2]$, and the Cauchy distribution are Gaussian mixtures. Let us mention that Costa \cite{costa} also considered symmetric stable laws to prove a strengthened version of the entropy power inequality that fails in general. 

Our first result proves the concavity of entropy conjectured in \cite{ball2016reverse} for Gaussian mixtures.

\begin{theorem} \label{2GMs}
Let $X_1,X_2$ be independent Gaussian mixtures. Then the function 
\begin{equation}
t \longmapsto h\Bigl(\sqrt{t}X_1 + \sqrt{1-t}X_2\Bigr)
\end{equation}
is concave on the interval $[0,1]$.
\end{theorem}

Theorem \ref{2GMs} will be a straightforward consequence of a more general result for the R\'enyi entropy of a weighted sum of $n$ Gaussian mixtures. Let $\triangle^{n-1}$ be the standard simplex in $\R^n$,
\begin{equation}
\triangle^{n-1} \eqdef \big\{(\pi_1,\ldots,\pi_n)\in[0,1]^n: \ \pi_1+\cdots+\pi_n=1\big\}.
\end{equation}
The R{\'e}nyi entropy of order $\alpha\neq 1$ of a random vector $X$ with density $f$ is given by
\begin{equation}
h_{\alpha}(X) \eqdef \frac{1}{1-\alpha}\log{\Bigl(\int_{\mathbb{R}^d}{f^{\alpha}(x)\,\diff x}\Bigr)},
\end{equation}
and $h_1(X)$ is simply the Shannon entropy $h(X)$. We will prove the following general concavity.

\begin{theorem} \label{nGMs}
Let $X_1,\ldots,X_n$ be independent Gaussian mixtures. Then, the function
\begin{equation} \label{eq:fun}
\triangle^{n-1} \ni (a_1^2,\ldots,a_n^2) \longmapsto h_{\alpha}\Big( \sum_{i=1}^n a_i X_i\Big)
\end{equation}
is concave on $\triangle^{n-1}$ for every $\alpha \geq 1$.
\end{theorem}

 When $n=2$ and $\alpha =1 $, Theorem \ref{nGMs} reduces exactly to Theorem \ref{2GMs}.

In \cite[Theorem~8]{eskeGM}, it was shown that if $X_1,\ldots,X_n$ are i.i.d., then the function \eqref{eq:fun} is \emph{Schur concave}, namely that if $(a_1,\ldots,a_n)$ and $(b_1,\ldots,b_n)$ are two unit vectors in $\mathbb{R}^n$, then 
\begin{equation} \label{eskegmth}
(a_1^2,\ldots,a_n^2) \preceq_\mr{m} (b_1^2,\ldots,b_n^2) \ \ \ \Longrightarrow \ \ \  h_{\alpha}\Bigl(\sum_{i=1}^n{a_iX_i}\Bigr) \geq h_{\alpha}\Bigl(\sum_{i=1}^n{b_iX_i}\Bigr),
\end{equation}
for any $\alpha \geq 1$, where $\preceq_\mr{m}$ is the majorisation ordering of vectors (see \cite{eskeGM}). As the unit vector with all coordinates equal to $\frac{1}{{n}}$ is majorised by any other vector in $\triangle^{n-1}$, \eqref{eskegmth} implies that the function \eqref{eq:fun} achieves its maximum on the main diagonal for Gaussian mixtures. 

As any permutationally invariant concave function is Schur concave (see \cite[p.~97]{MO}), \eqref{eskegmth} follows from Theorem \ref{nGMs}. On the other hand, the function $x_1\cdots x_n$ is  permutationally invariant and Schur concave on $\R_+^n$ (see \cite[p.~115]{MO}) but it is evidently not concave on the hyperplane $x_1+\cdots+x_n=1$ when $n\geq3$. Therefore, Theorem \ref{nGMs} is a strict refinement of \cite[Theorem~8]{eskeGM}.

We note in passing that, while the conclusion of Theorem \ref{2GMs} has been conjectured in \cite{ball2016reverse} to hold for every i.i.d.~log-concave random variables $X_1, X_2$, the conclusion of Theorem \ref{nGMs} cannot hold for this class of variables.  In \cite{MNT20}, Madiman, Nayar and Tkocz constructed a symmetric log-concave random variable $X$ for which the Schur concavity \eqref{eskegmth} does not hold for i.i.d.~copies of $X$ and thus, as a consequence of \cite[p.~97]{MO}, the concavity of Theorem \ref{nGMs} must also fail.




\subsection{Fisher Information}

Let $X$ be a continuous random vector in $\mathbb{R}^d$ with smooth density $f:\R^d\to\R_+$. The Fisher information of $X$ is the quantity 
\begin{equation}
I(X) \eqdef \int_{\mathbb{R}^d}{\frac{|\nabla f(x)|^2}{f(x)}\,\diff x} = \mb{E} \big[\big|\rho(X)\big|^2\big],
\end{equation}
where $\rho(x) \eqdef \frac{{\nabla f(x)}}{f(x)}$ is the {\em score function} of $X$. Fisher information and entropy are connected by the classical de Bruijn identity (see, e.g., \cite{Sta59}), due to which most results for Fisher information are formally stronger than their entropic counterparts. In particular, the inequality
\begin{equation} \label{eq:bs}
\forall \ t\in[0,1],\qquad \frac{1}{I(\sqrt{t}X_1 + \sqrt{1-t}X_2)} \geq \frac{t}{I(X_1)} + \frac{1-t}{I(X_2)}
\end{equation}
of Blachman and Stam \cite{Sta59,Bla65}, which holds for all independent random vectors $X_1,X_2$ in $\R^d$, implies the entropy power inequality  \eqref{eq:epi}. In the spirit of the question of Ball, Nayar and Tkocz \cite{ball2016reverse} and of the result of \cite{eskeGM}, we raise the following problem.

\begin{question} \label{q:fi}
Let $X_1,\ldots,X_n$ be i.i.d.~Gaussian mixtures. For which unit vectors $(a_1,\ldots,a_n)$ in $\R^n$ is the Fisher information of $\sum_{i=1}^n a_i X_i$ minimised?
\end{question}

While Question \ref{q:fi} still remains elusive, we shall now explain how to obtain some useful bounds for the Fisher information of mixtures. In order to state our results in the greatest possible generality, we consider random \emph{vectors} which are mixtures of centered multivariate Gaussians. Recall that the Fisher information matrix of a random vector $X$ in $\R^d$ is given by
\begin{equation}
\mathcal{I}(X)_{ij} \eqdef \int_{\mathbb{R}^d}{\frac{\partial_i f(x) \partial_j f(x)}{f(x)} \,\diff x},
\end{equation}
where $f:\R^d\to\R_+$ is the smooth density of $X$, so that $I(X) = \mr{tr}\mc{I}(X)$.

Let $\mathcal{F}_d \subset L_1(\mathbb{R}^d)$ be the space of smooth probability densities on $\mathbb{R}^d$. By abuse of notation, we will also write $I(f)$ and $\mathcal{I}(f)$ to denote the Fisher information and Fisher information matrix respectively of a random vector with smooth density $f$ on $\mathbb{R}^d$. In his recent treatise on estimates for the Fisher information, Bobkov made crucial use of the \emph{convexity} of the Fisher information functional $I(X)$ as a function of the density of the random variable $X$, see \cite[Proposition~15.2]{bobkov2022upper}. For our purposes we shall need the following matricial extension of this. 

\begin{proposition} \label{prop:conv}
Fix $d\in\N$. If $\pi$ is a Borel probability measure on $\mc{F}_d$, then
\begin{equation}  \label{FIJI}
\mc{I}\Big( \int_{\mc{F}_d} g \,\diff\pi(g)\Big) \preceq \int_{\mathcal{F}_d}{\mc{I}{(g)}\,\diff\pi(g)},
\end{equation}
provided that $\int_{\mathcal{F}_d}{\|\mc{I}{(g)\|_{\mr{op}}}\,\diff\pi(g)}<\infty$. Here $\preceq$ denotes the positive semi-definite ordering of matrices.
\end{proposition}

We propose the following definition of Gaussian mixtures in arbitrary dimension.

\begin{definition} \label{def}
A random vector $X$ in $\R^d$ is a (centered) Gaussian mixture if $X$ has the same distribution as ${\bf Y}Z$, where ${\bf Y}$ is a random $d\times d$ matrix which is almost surely positive definite (and thus symmetric) and $Z$ is a standard Gaussian random vector in $\R^d$, independent of ${\bf Y}$.\footnote{In view of the polar decomposition of ${\bf Y}$ as ${\bf Y} = ({\bf Y}^T{\bf Y})^{1/2}U$, where $U$ is an orthogonal matrix, and the rotational invariance of the Gaussian measure, the definition above could be rephrased with ${\bf Y}$ assumed to be almost surely nonsingular.  The nonsingularity assumption is imposed to ensure that $X$ has a density.}
\end{definition}

As in the scalar case, a Gaussian mixture $X$ in $\R^d$ has density of the form
\begin{equation} \label{eq:dens-GM}
\forall \ x\in\R^d, \qquad f_X(x) = \mb{E}\Big[ \frac{1}{\det(\sqrt{2\pi} {\bf Y})} e^{-|{\bf Y^{-1}}x|^2/2} \Big].
\end{equation}
Employing Proposition \ref{prop:conv} for Gaussian mixtures we deduce the following bound.

\begin{corollary} \label{cor:fisher-GM}
Fix $d\in\N$ and let $X$ be a random vector in $\R^d$ admitting a Gaussian mixture representation ${\bf Y}Z$. Then, we have
\begin{equation}
\mc{I}(X) \preceq \mb{E} \big[ ({\bf Y Y}^T)^{-1}\big].
\end{equation}
\end{corollary}

This upper bound should be contrasted with the general lower bound
\begin{equation} \label{eq:cr}
\mc{I}(X) \succeq \mr{Cov}(X)^{-1} = \big( \mb{E} {\bf  Y Y}^T\big)^{-1},
\end{equation}
where the first inequality is the multivariate Cr\'amer--Rao bound \cite[Theorem~3.4.4]{BD15}.



\subsubsection{Quantitative CLT for the Fisher information matrix of Gaussian mixtures}

Equality in the Cram\'er--Rao bound \eqref{eq:cr} is attained if and only if $X$ is Gaussian.  The deficit in the scalar version of this inequality is the relative Fisher information $I(X\| Z)$ between $X$ and $Z$ and may be interpreted as a strong measure of distance of $X$ from Gaussianity. In particular, in view of Gross' logarithmic Sobolev inequality \cite{gross} and Pinsker's inequality \cite{Pin64,csiszarpinskerref, kullback}, closeness in relative Fisher information implies closeness in relative entropy and a fortiori in total variation distance. 
Therefore, a very natural question is under which conditions and with what rate the relative Fisher information of a weighted sum tends to zero along the central limit theorem, thus offering a strenthening of the entropic central limit theorem \cite{barronentropy}. As an application of Corollary \ref{cor:fisher-GM}, we obtain a bound for a matrix analogue of the relative Fisher information of Gaussian mixtures. Here and throughout, $\|\cdot\|_{\mr{op}}$ denotes the operator norm of a square matrix.

\begin{theorem} \label{coropconvexity}
Fix $d\in\N$, $\delta\in(0,1]$ and let $X_1,\ldots,X_n$ be i.i.d.~random vectors in $\mathbb{R}^d$, each admitting a Gaussian mixture representation ${\bf Y}Z$ as above. Assume also that
\begin{equation} \label{matrixconditions}
\E\op{{\bf Y Y}^T}^{1+\delta}<\infty \quad \mbox{and} \quad \E\op{\big({\bf Y Y}^T\big)^{-1}}^{1+\delta} < \infty.
\end{equation}
Then, for every unit vector $a = (a_1,\ldots,a_n)$ in $\mathbb{R}^n$ the weighted sum $S_n = \sum_{i=1}^n{a_iX_i}$ satisfies 
\begin{equation} \label{eq:corop}
\op{\mathrm{Cov}(S_n)^{\frac{1}{2}}\mathcal{I}(S_n)\mathrm{Cov}(S_n)^{\frac{1}{2}}-\mathrm{I}_d} \leq C({\bf Y}) \log^\delta(d+1) \|a\|_{2+2\delta}^{\frac{2\delta}{1+\delta}},
\end{equation}
where $C({\bf Y})$ is a constant that depends only on the moments of $\|{\bf YY}^T\|_{\mr{op}}$.
\end{theorem}

There is a vast literature on quantitative versions of the central limit theorem. The first to obtain efficient bounds for the relative Fisher information of weighted sums were Artstein, Ball, Barthe and Naor \cite{artstein2004rate} (see also the work \cite{johnson2004fisher} of Johnson and Barron) who obtained a $O(\|a\|_4^4)$ upper bound on $I(S_n\| X)$, where $S_n = \sum_{i=1}^n a_i X_i$ for $X_1,\ldots, X_n$ i.i.d.~random variables satisfying a Poincar\'e inequality. In particular, this bound reduces to the sharp rate $O(\frac{1}{n})$ on the main diagonal. Following a series of works on the relative entropy of weighted sums \cite{bobkov2013rate,bobkov2014berry}, Bobkov, Chistyakov and G\"otze investigated in \cite{bobkov2014fisher} upper bounds for the relative Fisher information along the main diagonal under finite moment assumptions. More specifically, their main result asserts that if $\mb{E}|X_1|^s<\infty$ for some $s\in(2,4)$, then
\begin{equation} \label{bcg}
I\Bigl(\frac{1}{\sqrt{n}}\sum_{i=1}^n{X_i}\Big\|Z\Bigr) = O\left(\frac{1}{n^{\frac{s-2}{2}+o(1)}}\right),
\end{equation}
where the $n^{o(1)}$ term is a power of $\log n$, provided that the Fisher information of the sum is finite for some $n$. The exponent $\frac{s-2}{2}$ is sharp in this estimate. Moreover, it is also shown in \cite{bobkov2014fisher} that if $\mb{E}X_1^4<\infty$, then the relative Fisher information decays with the optimal $O(\frac{1}{n})$ rate of convergence. This is a far-reaching extension of the results of \cite{artstein2004rate, johnson2004fisher} on the main diagonal as the Poincar\'e inequality assumption in particular implies finiteness of all moments.


The scalar version of Theorem \ref{coropconvexity} (corresponding to $d=1$) is in various ways weaker than the results of \cite{bobkov2014fisher}. Firstly, it applies only within the class of Gaussian mixtures and it requires the finiteness of a \emph{negative} moment of the random variable besides a positive one. Moreover, even if these assumptions are satisfied, the bound \eqref{eq:corop} yields the rate $O(\frac{1}{n^{c_\delta}})$ with $c_\delta=\frac{\delta^2}{(1+\delta)^2}$ along the main diagonal if $X$ has a finite $2+2\delta$ moment. This is weaker than the sharp $O(\frac{1}{n^{\delta+o(1)}})$ which follows from \cite{bobkov2014fisher}. On the other hand, Theorem \ref{coropconvexity} applies to general coefficients beyond the main diagonal and, in contrast to \cite{artstein2004rate, johnson2004fisher}, does not require the finiteness of all positive moments. More importantly though, \eqref{eq:corop} is multi-dimensional bound with a \emph{subpolynomial} dependence on the dimension $d$. To the best of our knowledge, this is the first such bound for the relative Fisher information \emph{matrix} of a weighted sum and it would be very interesting to extend it to more general classes of random vectors and to obtain sharper rates.

The logarithmic dependence on the dimension in Theorem \ref{coropconvexity} is a consequence of a classical result of Tomczak-Jaegermann \cite{Tom74} on the uniform smoothness of Schatten classes. While Theorem \ref{coropconvexity} is stated in terms of the operator norm, the proof yields an upper bound for any operator monotone matrix norm (see Remark \ref{rem:norms}) in terms of its Rademacher type constants.


\subsection*{Acknowledgements} We are grateful to L\'eonard Cadilhac for helpful discussions.


\section{Concavity of entropy} \label{concavitysec}

This section is devoted to the proof of Theorem \ref{nGMs}. We shall make use of the standard variational formula for entropy which asserts that if $X$ is a continuous random variable, then
\begin{equation}
h(X) = \min\big\{ \mb{E}[ -\log g(X)]: \ g:\R\to\R_+ \mbox{ is a density function }\big\}.
\end{equation}

\begin{proof} [Proof of Theorem \ref{nGMs}]
We start with the Shannon entropy, which corresponds to $\alpha=1$. Fix two unit vectors $(a_1,\ldots,a_n)$ and $(b_1,\ldots,b_n)$ in $\R^n$. For $t\in[0,1]$, consider 
\begin{equation}
X_t \eqdef \sum_{i=1}^n \sqrt{t a_i^2+(1-t)b_i^2} X_i \qquad \mbox{and} \qquad f(t) \eqdef h(X_t),
\end{equation} 
and denote by $g_t:\R\to\R_+$ the density of $X_t$. The statement of the theorem is equivalent to the concavity of the function $f$ on the interval $[0,1]$.  

Let $\lambda,t_1,t_2 \in [0,1]$ and set $t = \lambda t_1 + (1-\lambda)t_2.$ By the variational formula for entropy, we have
\begin{equation} \label{in1}
\begin{split}
\lambda f(t_1) + (1-\lambda) f(t_2) = \lambda \mb{E} [ -\log g_{t_1}&(X_{t_1})] + (1-\lambda) \mb{E}[ -\log g_{t_2}(X_{t_2})] \\  
&\leq  \lambda\mb{E} [ -\log g_{t}(X_{t_1})] + (1-\lambda) \mb{E}[ -\log g_{t}(X_{t_2})] .
\end{split}
\end{equation}
Moreover, since $X_i$ has the same distribution as the independent product $Y_i Z_i$, the stability of Gaussian measure implies the equality in distribution
\begin{equation} \label{eq:dist}
X_t \stackrel{\mr{(d)}}{=} \sqrt{\sum_{i=1}^n \big( ta_i^2+(1-t)b_i^2\big) Y_i^2 } \ Z.
\end{equation}
Therefore, $X_t$ is itself a Gaussian mixture. By the characterisation of \cite[Theorem~2]{eskeGM}, this is equivalent to the complete monotonicity of the function $g_t(\sqrt{\cdot})$. Thus, by Bernstein's theorem, $g_t(\sqrt{\cdot})$ is the Laplace transform of a non-negative Borel measure on $(0,\infty)$ and therefore the function $\varphi_t \eqdef -\log{g_t(\sqrt{\cdot})}$ is concave on $(0,\infty)$. Hence, by \eqref{in1} and \eqref{eq:dist}, we have
\begin{equation}
\begin{split}
\lambda f(t_1) & + (1-\lambda) f(t_2) \\
& \leq \lambda \mb{E} \bigg[ \varphi_t \Big( \sum_{i=1}^n \big( t_1 a_i^2+(1-t_1)b_i^2\big) Y_i^2 Z^2\Big)\bigg] + (1-\lambda) \mb{E} \bigg[ \varphi_t \Big( \sum_{i=1}^n \big( t_2 a_i^2+(1-t_2)b_i^2\big) Y_i^2 Z^2\Big)\bigg] \\
& \leq \mb{E}\bigg[ \varphi_t\bigg( \sum_{i=1}^n \Big( \lambda\big( t_1 a_i^2+(1-t_1)b_i^2\big) + (1-\lambda)\big( t_2 a_i^2+(1-t_2)b_i^2\big) \Big) Y_i^2 Z^2 \bigg) \bigg] \\
& = \mb{E}\bigg[ \varphi_t\Big( \sum_{i=1}^n \big(  ta_i^2 + (1-t)b_i^2 \big) Y_i^2 Z^2 \Big) \bigg] = \mb{E}\bigg[ -\log g_t\Big( \sum_{i=1}^n \sqrt{t a_i^2 + (1-t)b_i^2} X_i \Big) \bigg] = f(t). \\
\end{split}
\end{equation}
This completes the proof of the concavity of Shannon entropy. 

Next, let $\alpha >1$ and consider again $t = \lambda t_1 + (1-\lambda)t_2$. Denoting by $\psi_t = g_t^{\alpha-1}(\sqrt{\cdot})$ and applying the same reasoning, we get
\begin{equation}
\begin{split}
\int_{\mathbb{R}}{g_t^{\alpha}(x)\, \diff x} &= \int_{\mathbb{R}}{g_t(x)g_t^{\alpha-1}(x) \,\diff x} = \mb{E} g_t^{\alpha-1}(X_t) \\
&= \mathbb{E}\bigg[\psi_t\bigg( \sum_{i=1}^n \Big( \lambda\big( t_1 a_i^2+(1-t_1)b_i^2\big) + (1-\lambda)\big( t_2 a_i^2+(1-t_2)b_i^2\big) \Big) Y_i^2 Z^2 \bigg)\bigg].
\end{split}
\end{equation}
Now $\psi_t = e^{-(\alpha-1)\varphi_t}$ is log-convex and thus 
\begin{equation} \label{lamdaholder}
\begin{split}
\int_{\mathbb{R}}{g_t^{\alpha}(x)\, \diff x} &\leq \mb{E}\bigg[ \psi_t^{\lambda}\Bigl(\sum_{i=1}^n\big( t_1a^2_i + (1-t_1)b^2_i\big) Y_i^2 Z^2 \Bigr)\psi_t^{1-\lambda}\Bigl(\sum_{i=1}^n \big(t_2a^2_i + (1-t^2_2)b_i^2\big)Y_i^2 Z^2\Bigr)\bigg] \\ 
& \leq \mb{E} \big[ g_t^{\alpha-1}(X_{t_1})\big]^\lambda \mb{E} \big[ g_t^{\alpha-1}(X_{t_2})\big]^{1-\lambda}
\end{split}
\end{equation} 
by H{\"o}lder's inequality and \eqref{eq:dist}. By two more applications of H{\"o}lder's inequality, we get
\begin{equation}\label{alphaholderlast0}
\int_{\mathbb{R}}{g_{t_1}(x)g_t(x)^{\alpha-1}\, \diff x} \leq \Bigl(\int_{\mathbb{R}}{g_{t_1}^{\alpha}(x)\, \diff x}\Bigr)^{\frac{1}{\alpha}}\Bigl(\int_{\mathbb{R}}{g_{t}^{\alpha}(x)\, \diff x}\Bigr)^{\frac{\alpha-1}{\alpha}} 
\end{equation}
and 
\begin{equation} \label{alphaholderlast}
\int_{\mathbb{R}}{g_{t_2}(x)g_t(x)^{\alpha-1}\, \diff x} \leq \Bigl(\int_{\mathbb{R}}{g_{t_2}^{\alpha}(x) \, \diff x}\Bigr)^{\frac{1}{\alpha}}\Bigl(\int_{\mathbb{R}}{g_{t}^{\alpha}(x)\, \diff x}\Bigr)^{\frac{\alpha-1}{\alpha}}.
\end{equation}
Combining \eqref{lamdaholder}, \eqref{alphaholderlast0} and \eqref{alphaholderlast} we thus obtain 
\begin{equation}
\Bigl(\int_{\mathbb{R}}{g_t^{\alpha}(x)\, \diff x}\Bigr)^{\frac{1}{\alpha}} \leq \Bigl(\int_{\mathbb{R}}{g_{t_1}^{\alpha}(x)\, \diff x}\Bigr)^{\frac{\lambda}{\alpha}}\Bigl(\int_{\mathbb{R}}{g_{t_2}^{\alpha}(x)\, \diff x}\Bigr)^{\frac{1-\lambda}{\alpha}}
\end{equation}
which is exactly the claimed concavity of R\'enyi entropy.
\end{proof}

\begin{remark}
One may wonder whether Theorem \ref{nGMs} can be extended to Gaussian mixtures on $\R^d$ in the sense of Definition \ref{def}.  Denoting by $\sqrt{M}$ the positive semidefinite square root of a positive semidefinite matrix $M$ and repeating the above argument, we would need the validity of the inequality
\begin{equation} \label{eq:want1}
\forall \ \lambda\in(0,1),\qquad g\big( \sqrt{\lambda A+(1-\lambda)B}z\big) \leq g\big(\sqrt{A}z\big)^\lambda g\big(\sqrt{B}z\big)^{1-\lambda}
\end{equation}
where $g:\R^d\to\R_+$ is the density of a Gaussian mixture,  A and B are positive semidefinite $d\times d$ matrices and $z$ is a vector in $\R^d$.  The validity of \eqref{eq:want1} for a Gaussian density with arbitrary covariance is equivalent to the operator concavity of the matrix function
\begin{equation}
f(X) \eqdef \sqrt{X}Y\sqrt{X}
\end{equation}
for an arbitrary positive semidefinite matrix $Y$. The following counterexample to this statement was communicated to us by L\'eonard Cadilhac. As the function $f$ takes values in the cone of positive semidefinite matrices, operator concavity is equivalent to operator monotonicity (see the proof of \cite[Theorem~V.2.5]{Bha97}). Take two non-negative matrices $A,Y$ such that $Y \preceq A$ but $Y^2 \npreceq A^2$. Then, the corresponding function $f(X) = \sqrt{X} Y\sqrt{X}$ satisfies $f(Y)=Y^2$ and $f(A) = \sqrt{A} Y \sqrt{A}  \preceq A^2$ since $Y \preceq A$. Therefore, $f(Y) \npreceq f(A)$ and thus $f$ is not operator monotone or concave.
\end{remark}


\section{Convexity of Fisher information} \label{fishersec}


\subsection{Warm-up: the Fisher information of independent products} Before showing the general argument which leads to Proposition \ref{prop:conv}, we present a short proof for the case of mixtures of dilates of a \emph{fixed} distribution which corresponds exactly to the Fisher information of a product of independent random variables. As this is a special case of Bobkov's \cite[Proposition~15.2]{bobkov2014fisher}, we shall disregard rigorous integrability assumptions for the sake of simplicity of exposition.

\begin{theorem} \label{logconcavegen}
Let $W$ be a random variable with zero mean and smooth-enough density and let $Y$ be an independent positive random variable. Then, 
\begin{equation}
\frac{1}{\mathbb{E}Y^2\Var{(W)}} \leq I(YW) \leq \mb{E}\Big[{\frac{I(W)}{Y^2}}\Big].
\end{equation}
\end{theorem}

\begin{proof}

The first inequality is the Cram{\'e}r-Rao lower bound. Suppose that $W$ has density $e^{-V}$ with $V$ nice enough. Then, $YW$ has density 
\begin{equation}
f(x) \eqdef \mb{E}\Big[\frac{1}{Y}e^{-V\big(\frac{x}{Y}\big)}\Big]
\end{equation}
and thus, differentiating under the expectation and using Cauchy--Schwarz, we get
\begin{equation}
\begin{split}
\p{f}(x)^2 &= \mb{E}\Big[\frac{\p{V}(\frac{x}{Y})}{Y^2}e^{-V\big(\frac{x}{Y}\big)}\Big]^2 \leq \mb{E}\Big[\frac{1}{Y}e^{-V\big(\frac{x}{Y}\bigl)}\Big] \mb{E}\Big[\frac{\p{V}(\frac{x}{Y})^2}{Y^3}e^{-V\big(\frac{x}{Y}\big)}\Big]
= f(x)\ \mb{E}\Big[\frac{\p{V}(\frac{x}{Y})^2}{Y^3}e^{-V\big(\frac{x}{Y}\big)}\Big].
\end{split}
\end{equation}
Thus, 
\begin{equation*}
\begin{split}
I(X) = \int_{\mathbb{R}}{\frac{\p{f}(x)^2}{f(x)}\, \diff x} & \leq \int_{\mathbb{R}}\mb{E}\Big[\frac{\p{V}(\frac{x}{Y})^2}{Y^3}e^{-V\big(\frac{x}{Y}\big)}\Big] \, \diff x = \mb{E}\Big[ \frac{1}{Y^2} \int_{\mathbb{R}} \frac{\p{V}(\frac{x}{Y})^2}{Y}e^{-V\big(\frac{x}{Y}\big)}\, \diff x \Big]
\\ & =  \mb{E}\Big[{\frac{1}{Y^2}}\Big] \mb{E}\Big[{\p{V}(W)^2}\Big]=  \mb{E}\Big[{\frac{I(W)}{Y^2}}\Big]. \qedhere
\end{split}
\end{equation*}
\end{proof}


\subsection{Proof of Proposition \ref{prop:conv}} \label{dimensionsec}

\newcommand{\la}{\lambda}

We start by proving the two-point convexity of $\mc{I}$.

\begin{proposition} \label{opconvex}
The Fisher information matrix is operator convex on $\mathcal{F}_d$, that is, for $f_1,f_2 \in \mathcal{F}_d$,
\begin{equation} \label{f1f2convexityI}
\forall \ \theta\in[0,1],\qquad \mathcal{I}\big(\theta f_1+(1-\theta)f_2\big) \preceq \theta\mathcal{I}(f_1) + (1-\theta)\mathcal{I}(f_2).
\end{equation}
\end{proposition}
\begin{proof}
First we claim that the function $R:\mathbb{R}^d \times \mathbb{R}_{+} \rightarrow \mathbb{R}^{d \times d}$, given by $R(x,\lambda) = \frac{xx^T}{\lambda}$
is jointly operator convex. To prove this, we need to show that for every $\theta \in (0,1),$ $x,y \in \mathbb{R}^d$ and $\la,\mu > 0$,
\begin{equation} \label{Rconvexity}
R(\theta x + (1-\theta)y, \theta \la + (1-\theta)\mu) \preceq \theta R(x,\la) + (1-\theta) R(y,\mu). 
\end{equation}
After rearranging, this can be rewritten as 
\begin{equation}
\theta(1-\theta)\bigl(\la^2 xx^T +\mu^2 yy^T- \la\mu xy^T -\la\mu yx^T \bigr) \succeq 0,
\end{equation}
which is true since it is equivalent to $(\lambda x - \mu y) ( \lambda x - \mu y)^T \succeq 0$. 

Since the Fisher information matrix can be written as
\begin{equation}
\mc{I}(f) = \int_{\R^d} R\big( \nabla f(x), f(x)\big) \,\diff x,
\end{equation}
the conclusion follows by the convexity of $R$ and the linearity of $\nabla$ and $\int$.
\end{proof}

In order to derive the general Jensen inequality of Proposition \ref{prop:conv} from Proposition \ref{opconvex}, we will use a somewhat involved compactness argument that was invoked in \cite{bobkov2014fisher,bobkov2022upper}. We point out that these intricacies arise since the space $\mc{F}_d$ of smooth densities in $\R^d$ is infinite-dimensional. As our argument shares similarities with Bobkov's, we shall only point out the necessary modifications which need to be implemented. We stard by proving the following technical lemma.

\begin{lemma} \label{lemmaquadsemi}
Let $X, \{X_k\}_{k\geq 1}$ be random vectors in $\mathbb{R}^d$ such that $X_k \Rightarrow X$ weakly.
\begin{enumerate} [label =(\roman*)]
 \item \label{lowersemifirstpart} If $\sup_k{\op{\mathcal{I}(X_k)}} < \infty$, then for every $x\in\mb{S}^{d-1}$,
\begin{equation} \label{qformlowersemicont}
\ix{X} \leq \liminf_{k \to \infty}{\ix{X_k}}.
\end{equation}
\item Moreover, we always have
\begin{equation} \label{oplowersemi}
\op{\mathcal{I}(X)} \leq \liminf_{k \to \infty}{\op{\mathcal{I}(X_k)}}.
\end{equation}
\end{enumerate}
\end{lemma}

\begin{proof}
We start with \eqref{qformlowersemicont}. It clearly suffices to show that any subsequence of $\{X_k\}$ has a further subsequence for which the conclusion holds. If $\op{\mathcal{I}(X_k)} \leq I < \infty$ for all $k\geq1$, then
\begin{equation}
I(X_k) = \mathrm{tr}(\mathcal{I}(X_k)) \leq d\op{\mathcal{I}(X_k)} \leq dI < \infty.
\end{equation}
Write $f_k$ and $f$ for the densities of $X_k$ and $X$ respectively. Choose and fix any subsequence of $\{f_k\}$. 
By the proof of \cite[Proposition 14.2]{bobkov2022upper}, using the boundedness of Fisher informations, there is a further subsequence, say $f_{k_j}$, for which $f_{k_j} \to f$ and $\nabla f_{k_j} \to \nabla f$ a.e. as $j \to \infty$. Therefore 
\begin{equation}
\lim_{j \to \infty}{\Bigl\langle \frac{\nabla f_{k_j}(u)\nabla f_{k_j}(u)^T}{f_{k_j}(u)}x,x\Bigr\rangle \mathbb{I}_{\{f_{k_j}(u) > 0\}}} = \Bigl\langle\frac{\nabla f(u)\nabla f(u)^T}{f(u)}x,x\Bigr\rangle\mathbb{I}_{\{f(u) > 0\}} 
\end{equation}
for almost every $u$. Integration with respect to $u$, linearity and Fatou's lemma yield \eqref{qformlowersemicont}.

To prove \eqref{oplowersemi}, fix a subsequence $X_{k_j}$ for which the $\liminf$ in \eqref{oplowersemi} is attained and without loss of generality assume that it is finite.  Then the subsequence satisfies $\sup_j \|\mc{I}(X_{k_j})\|_{\mr{op}} <\infty$ and thus by \eqref{qformlowersemicont} for every $x\in\mb{S}^{d-1}$ we have
\begin{equation}
\ix{X} \leq \liminf_{j \to \infty}{\ix{X_{k_j}}} \leq \liminf_{j\to\infty} \big\|\mc{I}(X_{k_j})\big\|_{\mr{op}} = \liminf_{k \to \infty}{\op{\mathcal{I}(X_k)}}.
\end{equation}
Taking a supremum over $x\in\mb{S}^{d-1}$ concludes the proof as $\mc{I}(X)$ is positive semi-definite. 
\end{proof}

Equipped with the lower semi-continuity of $\mc{I}$, we proceed to the main part of the proof.

\begin{proof} [Proof of Proposition \ref{prop:conv}]
Inequality \eqref{f1f2convexityI} may be extended to arbitrary finite mixtures by induction, that is 
if $p_1,\ldots,p_N\geq0$ satisfy $\sum_{i=1}^N{p_i} = 1,$ then
\begin{equation} \label{jensenfinitemixtures}
\mathcal{I}\Big( \sum_{i=1}^N p_if_i\Big) \preceq \sum_{i=1}^N{p_i\mathcal{I}(f_i)}.
\end{equation}
We need to extend \eqref{jensenfinitemixtures} to arbitrary mixtures. We write $\mathcal{F}_d(I) =\{f \in \mathcal{F}_d: \op{\mathcal{I}(f)} \leq I\}$ and $\mathcal{F}_d(\infty) = \cup_{I}{\mathcal{F}_d(I)}$. By the assumption $\int_{\mc{F}_d} \|\mc{I}(g)\|_{\mr{op}} \,\diff\pi(g) <\infty$, we deduce that the measure $\pi$ is supported on $\mathcal{F}_d(\infty)$. We shall prove that 
\begin{equation} \label{sufficesconvexityofquadr}
\forall x\in\mb{S}^{d-1},\qquad \left\langle \mc{I}\Big( \int_{\mc{F}_d} g\,\diff \pi(g) \Big)x,x \right\rangle \leq \int_{\mathcal{F}_d}{\ix{g}\, \diff\pi(g)}.
\end{equation}
Fix $x \in \mathbb{S}^{d-1}$ and $I\in\N$. By the operator convexity of the Fisher information matrix (Proposition \ref{opconvex}), the functional 
\begin{equation}
f \rightarrow \ix{f}
\end{equation}
is convex and by Lemma \ref{lemmaquadsemi} lower semi-continuous on $\mathcal{F}_d(I)$. Again by operator convexity, the set 
$\mathcal{F}_d(I) $  is convex and by Lemma \ref{lemmaquadsemi} it is closed. Now we may repeat exactly the same proof as in \cite[Proposition 15.1, Steps 1-2]{bobkov2022upper},
but working with the functional $\ix{f}$ instead of the Fisher information $I(f)$, to obtain \eqref{sufficesconvexityofquadr} if the measure $\pi$ is supported on $\mc{F}_d(I)$.

To derive inequality \eqref{sufficesconvexityofquadr} in general, fix $I_0$ large enough such that $\pi(\mc{F}_d(I_0))>\frac{1}{2}$ and for $I\geq I_0$ write the inequality \eqref{sufficesconvexityofquadr} for the restriction of $\pi$ to $\mc{F}_d(I)$, namely
\begin{equation} \label{eq:47}
\left\langle \mc{I}\bigg( \frac{1}{\pi(\mc{F}_d(I))} \int_{\mc{F}_d(I)} g\, \diff \pi(g) \bigg)x,x \right\rangle \leq  \frac{1}{\pi(\mc{F}_d(I))}  \int_{\mathcal{F}_d(I)}{\ix{g}\, \diff\pi(g)}.
\end{equation}
Denoting by $f_I$ the density on the left-hand side of the inequality, we have that $f_I$ converges weakly to the density $\int_{\ms{F}_d} g\,\diff\pi(g)$ as $I\to\infty$ and moreover \eqref{eq:47} yields
\begin{equation}
\forall \ I\geq I_0,\qquad \big\| \mc{I}(f_I)\big\|_{\mr{op}} \leq \frac{1}{\pi(\mc{F}_d(I))} \int_{\mc{F}_d(I)} \big\| \mc{I}(g)\big\|_{\mr{op}}\,\diff\pi(g) \leq 2  \int_{\mc{F}_d} \big\| \mc{I}(g)\big\|_{\mr{op}}\,\diff\pi(g) <\infty.
\end{equation}
Therefore, the assumptions of \eqref{qformlowersemicont} are satisfied for $\{f_I \}_{I\geq I_0}$ and thus
\begin{equation} \begin{split}
\left\langle \mc{I}\Big( \int_{\mc{F}_d} g\,\diff \pi(g) \Big)x,x \right\rangle & \leq \liminf_{I\to\infty} \left\langle \mc{I}\bigg( \frac{1}{\pi(\mc{F}_d(I))} \int_{\mc{F}_d(I)} g\, \diff \pi(g) \bigg)x,x \right\rangle 
\\ & \stackrel{\eqref{eq:47}}{\leq} \liminf_{I\to\infty} \frac{1}{\pi(\mc{F}_d(I))}  \int_{\mathcal{F}_d(I)}{\ix{g}\, \diff\pi(g)} = \int_{\mathcal{F}_d}{\ix{g}\, \diff\pi(g)},
\end{split}
\end{equation}
and this concludes the proof.
\end{proof}

\begin{proof} [Proof of Corollary \ref{cor:fisher-GM}]
In view of \eqref{eq:dens-GM} and Proposition \ref{prop:conv}, we have
\begin{equation}
\mathcal{I}({\bf Y}Z) = \mathcal{I}\bigg( \mb{E}_{{\bf Y}}\Big[ \frac{1}{\det(\sqrt{2\pi} {\bf Y})} e^{-|{\bf Y^{-1}} \cdot \ |^2/2} \Big] \bigg) \preceq  \mb{E}_{{\bf Y}} \big[\mc{I}({\bf Y}Z) \big]= \mathbb{E}\big[ ({\bf YY}^T)^{-1}\big],
\end{equation}
since the Fisher information matrix of a Gaussian vector with covariance matrix $\Sigma$ is $\Sigma^{-1}$.
\end{proof}


\section{CLT for the Fisher information matrix}

Before delving into the proof of Theorem \ref{coropconvexity}, we shall discuss some geometric preliminaries. Recall that a normed space $(V,\|\cdot\|_V)$ has Rademacher type $p\in[1,2]$ with constant $T\in(0,\infty)$ if for every $n\in\N$ and every $v_1,\ldots,v_n \in V$, we have
\begin{equation} \label{eq:def-typ}
\frac{1}{2^n} \sum_{\e \in\{-1,1\}^n} \Big\| \sum_{i=1}^n \e_i v_i\Big\|_V^p \leq T^p \sum_{i=1}^n \|v_i\|_V^p.
\end{equation}
The least constant $T$ for which this inequality holds will be denoted by $\Tsf_p(V)$. A standard symmetrisation argument (see, for instance, \cite[Proposition~9.11]{LT91}) shows that for any $n\in\N$ and any $V$-valued random vectors $V_1,\ldots,V_n$ with $\mb{E}[V_i]=0$ we have
\begin{equation} \label{eq:lt}
\mb{E}\Big\| \sum_{i=1}^n V_i\Big\|_V^p \leq \big( 2\Tsf_p(V)\big)^p \sum_{i=1}^n \mb{E}\|V_i\|_V^p.
\end{equation}

We denote by $\M_d(\R)$ the vector space of all $d\times d$ matrices with real entries. We shall consider the $p$-Schatten trace class $\msf{S}_p^d$ of $d\times d$ matrices. This is the normed space $\msf{S}_p^d = (\M_d(\R),\|\cdot\|_{\msf{S}_p})$, where for a $d\times d$ real matrix $A$, we denote
\begin{equation}
\|A\|_{\msf{S}_p} \eqdef \Big( \sum_{i=1}^d \sigma_i(A)^p\Big)^{1/p}
\end{equation}
and by $\sigma_1(A) \geq \cdots \geq \sigma_d(A)$ the singular values of $A$. Evidently, $\|\cdot\|_{\mr{op}} = \|\cdot\|_{\msf{S}_\infty}$. A classical result of Tomczak-Jaegermann \cite{Tom74} (see also \cite{BCL94} for the exact values of the constants) asserts that if $p\in[1,2]$, then $\msf{S}_p^d$ has Rademacher type $p$ constant $T_p\big(\msf{S}_p^d\big)=1$ and if $p\geq2$, then $\msf{S}_p^d$ has Rademacher type 2 constant $\Tsf_2\big(\msf{S}_p^d\big) \leq \sqrt{p-1}$. We shall use the following consequence of this.

\begin{lemma} \label{lemmalogd}
Fix $n,d\in\N$ and let $W_1,\ldots,W_n$ be i.i.d.~random $d\times d$ matrices with $\mb{E}[W_i]=0$. For any $\delta\in(0,1]$ and any vector $b=(b_1,\ldots,b_n)\in\R^n$, we have
\begin{equation} \label{p>2}
p\in[2,\infty) \qquad \Longrightarrow \qquad \mb{E}\Big\|{\sum_{i=1}^n b_i {W_i}}\Big\|_{\msf{S}_p}^{1+\delta} \leq 2^{1+\delta} (p-1)^\delta \mb{E}\big[\|W_1\|_{\msf{S}_p}^{1+\delta}\big]\ \|b\|_{1+\delta}^{1+\delta}
\end{equation}
and
\begin{equation} \label{p<2}
p\in[1+\delta,2] \qquad \Longrightarrow \qquad \mb{E}\Big\|{\sum_{i=1}^n b_i {W_i}}\Big\|_{\msf{S}_p}^{1+\delta} \leq 2^{1+\delta} \mb{E}\big[\|W_1\|_{\msf{S}_p}^{1+\delta}\big]\ \|b\|_{1+\delta}^{1+\delta}.
\end{equation}
Moreover,
\begin{equation} \label{opop}
\mb{E}\Big\|{\sum_{i=1}^n b_i {W_i}}\Big\|_{\mr{op}}^{1+\delta} \leq (2e)^{1+\delta} \log^\delta (d+1) \mb{E}\big[\|W_1\|_{\mr{op}}^{1+\delta}\big]\ \|b\|_{1+\delta}^{1+\delta}.
\end{equation}
\end{lemma}

\begin{proof}
We first prove \eqref{p>2}. In view of inequality \eqref{eq:lt}, it suffices to prove that the Rademacher type $(1+\delta)$-constant of $\msf{S}_p^d$ satisfies $\Tsf_{1+\delta}(\msf{S}_p^d) \leq (p-1)^{\frac{\delta}{1+\delta}}$. Given a  normed space $(X,\|\cdot\|_X)$ and $n\in\N$, consider the linear operator $T_n:\ell_p^n(X) \to L_p(\{-1,1\}^n;X)$ given by
\begin{equation}
\forall \ x=(x_1,\ldots,x_n)\in \ell_p^n(X), \qquad [T_nx](\e) = \sum_{i=1}^n \e_i x_i,
\end{equation}
where $\e=(\e_1,\ldots,\e_n)\in\{-1,1\}^n$. Then, it follows from \eqref{eq:def-typ} that
\begin{equation}
T_p(X) =  \sup_{n\in\N}\ \big\|T_n\big\|_{\ell_p^n(X)\to L_p(\{-1,1\}^n;X)}.
\end{equation}
In fact, if $X$ is finite-dimensional (like $\msf{S}_p^d$) then it was shown in \cite[Lemma~6.1]{FLM77} that the supremum is attained for some $n\leq \mr{dim}(X)(\mathrm{dim}(X)+1)/2$. Either way, by complex interpolation of vector-valued $L_p$ spaces (see \cite[Section~5.6]{BL76}), we thus deduce that
\begin{equation}
\Tsf_{1+\delta}\big( \msf{S}_p^d\big) \leq \msf{T}_1\big(\msf{S}_p^d\big)^\theta \msf{T}_2\big(\msf{S}_p^d\big)^{1-\theta},
\end{equation}
where $\frac{\theta}{1}+\frac{1-\theta}{2} = \frac{1}{1+\delta}$. The conclusion of \eqref{p>2} follows by plugging-in the value of $\theta$ and the result of \cite{Tom74,BCL94}. The proof of inequality \eqref{p<2} is similar, interpolating between 1 and $p$.

Finally, to deduce \eqref{opop}, note that for any $A\in\M_d(\R)$,
\begin{equation}
\|A\|_{\mr{op}} \leq \|A\|_{\msf{S}_p} \leq d^{1/p} \|A\|_{\mr{op}}
\end{equation}
and thus plugging $p=\log (d+1)+1$ in \eqref{p>2} we derive the desired inequality.
\end{proof}

Equipped with these inequalities, we can now proceed to the main part of the proof.

\begin{proof} [Proof of Theorem \ref{coropconvexity}]
Since $\E{S_n} = 0$ and $\mathrm{Cov}(S_n) = \E{{\bf YY}^T}$, we have
\begin{equation} \label{StandFImatrixfirst}
\op{\mathrm{Cov}(S_n)^{\frac{1}{2}}\mathcal{I}(S_n)\mathrm{Cov}(S_n)^{\frac{1}{2}}-\mathrm{I}_d} \leq \op{\E{{\bf YY}^T}}\op{\mathcal{I}(S_n)-\bigl(\E{{\bf YY}^T}\bigr)^{-1}},
\end{equation}
using that for any PSD matrices $A,B,$ $\op{AB} \leq \op{A}\op{B}$ and $\op{A^{\frac{1}{2}}} = \op{A}^{\frac{1}{2}}$. Now, $S_n$ is a Gaussian mixture itself and it satisfies
\begin{equation}
S_n = \sum_{i=1}^n a_i {\bf Y}_i Z_i \stackrel{\mr{(d)}}{=} \Big( \sum_{i=1}^n a_i^2 {\bf Y}_i {\bf Y}_i^T\Big)^{1/2} Z,
\end{equation}
Corollary \ref{cor:fisher-GM} yields the estimate
\begin{equation} \label{FImatrixupperbound}
\mathcal{I}(S_n)  \preceq \E{\Bigl(\sum_{i=1}^n{a_i^2{\bf Y}_i{\bf Y}_i^T}\Bigr)^{-1}}.
\end{equation}
Moreover by the multivariate Cram{\'e}r-Rao lower bound \cite[Theorem~3.4.4]{BD15}, we have 
\begin{equation} \label{mvcr}
\mathcal{I}(S_n) \succeq \bigl(\E{ {\bf YY}^T}\bigr)^{-1}
\end{equation}
and thus the matrix in the right-hand side of \eqref{StandFImatrixfirst} is positive semi-definite. Therefore, since $\|\cdot\|_{\mr{op}}$ is increasing with respect to the matrix ordering on positive matrices, \eqref{FImatrixupperbound} and \eqref{mvcr} yield
\begin{equation}
\op{\mathcal{I}(S_n)-\bigl(\E{{\bf YY}^T}\bigr)^{-1}} \leq \bigg\| \E{\Bigl(\sum_{i=1}^n{a_i^2{\bf Y}_i{\bf Y}_i^T}\Bigr)^{-1}} - \bigl(\E{ {\bf YY}^T}\bigr)^{-1} \bigg\|_{\mr{op}}.
\end{equation}

For $i=1,\ldots,n$ consider the i.i.d.~random matrices $W_i \eqdef {\bf Y}_i{\bf Y}_i^T - \E{{\bf YY}^T}$ and denote the event $\Ee \eqdef \bigl\{\op{\sum_{i=1}^n{a_i^2W_i}} \leq \e\bigr\}$. To bound the probability of the complement of $\Ee$, notice that
\begin{equation}
\begin{split}
\mathbb{P}{\{\Ee^{\mathrm{c}}\}} = \mb{P}\bigg\{\Big\|{\sum_{i=1}^n{a_i^2W_i}}\Big\|_{\mr{op}}^{1+\delta} > \e^{1+\delta}\bigg\} & \leq \frac{1}{\e^{1+\delta}}\E{\Big\|{\sum_{i=1}^n{a_i^2W_i}}\Big\|_{\mr{op}}^{1+\delta}} \\ & \stackrel{\eqref{opop}}{\leq} \Big(\frac{2e}{\e}\Big)^{1+\delta} \log^\delta (d+1) \mb{E}\big[\|W_1\|_{\mr{op}}^{1+\delta}\big]\ \big\|a\big\|_{2+2\delta}^{2+2\delta}.
\end{split}
\end{equation}
Moreover, since $\mb{E}\|W_1\|_{\mr{op}}^{1+\delta} \leq 2^{1+\delta} \mb{E} \|{\bf YY}^T\|_{\mr{op}}^{1+\delta}$,we get the bound 
\begin{equation} \label{eq:bp}
\mb{P}\{\Ee^\mr{c}\} \leq \Big(\frac{4e}{\e}\Big)^{1+\delta} \log^\delta (d+1) \mb{E}\big[\|{\bf YY^T}\|_{\mr{op}}^{1+\delta}\big]\ \big\|a\big\|_{2+2\delta}^{2+2\delta}.
\end{equation}
Next, we write 
\begin{equation} \label{twoterms}
\E{\Bigl(\sum_{i=1}^n{a_i^2{\bf Y}_i{\bf Y}_i^T}\Bigr)^{-1}} = \E\bigg[{\Bigl(\sum_{i=1}^n{a_i^2 {\bf Y}_i{\bf Y}_i^T}\Bigr)^{-1}\ \mb{I}_{\Ee}}\bigg] + \E\bigg[{\Bigl(\sum_{i=1}^n{a_i^2 {\bf Y}_i{\bf Y}_i^T}\Bigr)^{-1}\mb{I}_{\Ee^\mr{c}}}\bigg]
\end{equation}
and use the triangle inequality to get
\begin{equation} \label{eq:3terms}
\begin{split}
\bigg\| \E{\Bigl(\sum_{i=1}^n{a_i^2{\bf Y}_i{\bf Y}_i^T}\Bigr)^{-1}} - \bigl(\E{ {\bf YY}^T}\bigr)^{-1} \bigg\|_{\mr{op}} & \leq \Bigg\| \E\bigg[{\bigg(\Bigl(\sum_{i=1}^n{a_i^2{\bf Y}_i{\bf Y}_i^T}\Bigr)^{-1}} - \bigl(\E{ {\bf YY}^T}\bigr)^{-1}\bigg) \mb{I}_{\Ee} \bigg]\Bigg\|_{\mr{op}} 
\\ & + \mb{P}\{\Ee^\mr{c}\}\ \big\| \big( \mb{E} {\bf YY}^T\big)^{-1}\big\|_{\mr{op}}  +  \Bigg\| \E\bigg[{\Bigl(\sum_{i=1}^n{a_i^2{\bf Y}_i{\bf Y}_i^T}\Bigr)^{-1}} \mb{I}_{\Ee^\mr{c}}\bigg] \Bigg\|_{\mr{op}}.
\end{split}
\end{equation}

To control the first term in \eqref{eq:3terms}, we use Jensen's inequality for $\|\cdot\|_{\mr{op}}$ to get
\begin{equation}
\begin{split}
 \Bigg\| \E\bigg[{\bigg(\Bigl(\sum_{i=1}^n{a_i^2{\bf Y}_i{\bf Y}_i^T}\Bigr)^{-1}} & - \bigl(\E{ {\bf YY}^T}\bigr)^{-1}\bigg) \mb{I}_{\Ee} \bigg]\Bigg\|_{\mr{op}} \leq  \E\bigg[ \bigg\|{\Bigl(\sum_{i=1}^n{a_i^2{\bf Y}_i{\bf Y}_i^T}\Bigr)^{-1}} - \bigl(\E{ {\bf YY}^T}\bigr)^{-1}\bigg\|_{\mr{op}}\mb{I}_{\Ee} \bigg]
 \\ & \leq \big\| \big( \mb{E} {\bf YY}^T \big)^{-1}\big\|_{\mr{op}} \mb{E}\bigg[ \bigg\|{\Bigl(\sum_{i=1}^n{a_i^2{\bf Y}_i{\bf Y}_i^T}\Bigr)^{-1}} \bigg\|_{\mr{op}}\ \bigg\|\sum_{i=1}^n{a_i^2{\bf Y}_i{\bf Y}_i^T} - \E{ {\bf YY}^T}\bigg\|_{\mr{op}} \mb{I}_{\Ee} \bigg]
 \\ & =  \big\| \big( \mb{E} {\bf YY}^T \big)^{-1}\big\|_{\mr{op}} \mb{E}\bigg[ \bigg\|{\Bigl(\sum_{i=1}^n{a_i^2{\bf Y}_i{\bf Y}_i^T}\Bigr)^{-1}} \bigg\|_{\mr{op}}\ \bigg\|\sum_{i=1}^n{a_i^2 W_i}\bigg\|_{\mr{op}} \mb{I}_{\Ee} \bigg],
\end{split}
\end{equation}
where the second line follows from the identity $X^{-1}-Y^{-1} = X^{-1}(Y-X)Y^{-1}$ which yields the inequality $\|X^{-1}-Y^{-1}\|_{\mr{op}} \leq \|X^{-1}\|_{\mr{op}} \|Y^{-1}\|_{\mr{op}} \|X-Y\|_{\mr{op}}$ for positive matrices $X,Y$. Now, by the definition of the event $\Ee$ the last factor is at most $\e$ and thus we derive the bound
\begin{equation} \label{eq:almost1}
 \Bigg\| \E\bigg[{\bigg(\Bigl(\sum_{i=1}^n{a_i^2{\bf Y}_i{\bf Y}_i^T}\Bigr)^{-1}} - \bigl(\E{ {\bf YY}^T}\bigr)^{-1}\bigg) \mb{I}_{\Ee} \bigg]\Bigg\|_{\mr{op}} \leq \big\| \big( \mb{E} {\bf YY}^T \big)^{-1}\big\|_{\mr{op}} \mb{E}\bigg[ \bigg\|{\Bigl(\sum_{i=1}^n{a_i^2{\bf Y}_i{\bf Y}_i^T}\Bigr)^{-1}} \bigg\|_{\mr{op}} \bigg] \ \e.
\end{equation}
Finally, the function $A\mapsto A^{-1}$ is operator convex on positive matrices (see \cite[p.~117]{Bha97}), thus
\begin{equation} \label{eq:inv}
\Bigl(\sum_{i=1}^n{a_i^2{\bf Y}_i{\bf Y}_i^T}\Bigr)^{-1} \preceq \sum_{i=1}^n a_i^2 \big( {\bf Y}_i {\bf Y}_i^T \big)^{-1} \qquad \mbox{and} \qquad \big( \mb{E} {\bf YY}^T \big)^{-1} \preceq  \mb{E} \big({\bf YY}^T \big)^{-1}.
\end{equation}
Applying the operator norm on both sides, plugging this in \eqref{eq:almost1} and using the triangle inequality after taking the expectation, we conclude that
\begin{equation} \label{eq:term1}
\begin{split}
 \Bigg\| \E\bigg[{\bigg(\Bigl(\sum_{i=1}^n{a_i^2{\bf Y}_i{\bf Y}_i^T}\Bigr)^{-1}} - \bigl(\E{ {\bf YY}^T}\bigr)^{-1}\bigg) \mb{I}_{\Ee} \bigg]\Bigg\|_{\mr{op}} \leq \big(\mb{E}\big\| \big( {\bf YY}^T \big)^{-1}\big\|_{\mr{op}}\big)^2 \ \e.
\end{split}
\end{equation}
In view of \eqref{eq:bp} and \eqref{eq:inv}, the second term in \eqref{eq:3terms} is bounded by
\begin{equation} \label{eq:term2}
 \mb{P}\{\Ee^\mr{c}\}\ \big\| \big( \mb{E} {\bf YY}^T\big)^{-1}\big\|_{\mr{op}} \leq \Big(\frac{4e}{\e}\Big)^{1+\delta} \log^\delta (d+1) \mb{E}\big\|{\bf YY^T}\big\|_{\mr{op}}^{1+\delta}\ \mb{E}\big\| \big( {\bf YY}^T\big)^{-1}\big\|_{\mr{op}}  \ \big\|a\big\|_{2+2\delta}^{2+2\delta}.
\end{equation}
To bound the third term in \eqref{eq:3terms}, we use Jensen's inequality and \eqref{eq:inv} to get
\begin{equation} \label{for3term}
\begin{split}
 \Bigg\| \E\bigg[{\Bigl(\sum_{i=1}^n{a_i^2{\bf Y}_i{\bf Y}_i^T}\Bigr)^{-1}} \mb{I}_{\Ee^\mr{c}}\bigg] \Bigg\|_{\mr{op}} & \leq  \E\bigg[ \bigg\|{\Bigl(\sum_{i=1}^n{a_i^2{\bf Y}_i{\bf Y}_i^T}\Bigr)^{-1}}  \bigg\|_{\mr{op}}\mb{I}_{\Ee^\mr{c}}\bigg] 
 \\ & \stackrel{\eqref{eq:inv}}{\leq} \E\bigg[ \bigg\|{\sum_{i=1}^n{a_i^2 \big({\bf Y}_i{\bf Y}_i^T\big)^{-1}}}  \bigg\|_{\mr{op}}\mb{I}_{\Ee^\mr{c}}\bigg]
   \leq\E\bigg[ \Big({\sum_{i=1}^n{a_i^2 \big\|\big({\bf Y}_i{\bf Y}_i^T\big)^{-1}\big\|_{\mr{op}}}}\Big)  \mb{I}_{\Ee^\mr{c}}\bigg]
 \end{split}
\end{equation}
where the last estimate follows from the triangle inequality. Now, by H{\"o}lder's inequality,
\begin{equation}
\begin{split}
\E\bigg[ \Big({\sum_{i=1}^n{a_i^2 \big\|\big({\bf Y}_i{\bf Y}_i^T\big)^{-1}\big\|_{\mr{op}}}}\Big)  \mb{I}_{\Ee^\mr{c}}\bigg] & \leq  \E\bigg[ \Big( \sum_{i=1}^n a_i^2 \big\| \big( {\bf Y}_i{\bf Y}_i^T\big)^{-1} \big\|_{\mr{op}} \Big)^{1+\delta} \bigg]^{\frac{1}{1+\delta}} \mb{P}\{\Ee^\mr{c}\}^{\frac{\delta}{1+\delta}}
\\ & \leq \big(\mb{E} \big\| \big( {\bf YY}^T\big)^{-1}\big\|_{\mr{op}}^{1+\delta} \big)^{\frac{1}{1+\delta}} \mb{P}\{\Ee^\mr{c}\}^{\frac{\delta}{1+\delta}},
\end{split}
\end{equation}
where the last line follows from the triangle inequality in $L_{1+\delta}$. Combining this with \eqref{for3term} and \eqref{eq:bp} we thus conclude that
\begin{equation*} \label{eq:term3}
\Bigg\| \E\bigg[{\Bigl(\sum_{i=1}^n{a_i^2{\bf Y}_i{\bf Y}_i^T}\Bigr)^{-1}} \mb{I}_{\Ee^\mr{c}}\bigg] \Bigg\|_{\mr{op}} \leq  \Big(\frac{4e}{\e}\Big)^{\delta} \log^{\frac{\delta^2}{1+\delta}} (d+1) \big( \mb{E}\big\|{\bf YY^T}\big\|_{\mr{op}}^{1+\delta}\big)^{\frac{\delta}{1+\delta}}  \big(\mb{E} \big\| \big( {\bf YY}^T\big)^{-1}\big\|_{\mr{op}}^{1+\delta} \big)^{\frac{1}{1+\delta}}  \ \big\|a\big\|_{2+2\delta}^{2\delta}.
\end{equation*}
Plugging this bound along with \eqref{eq:term1} and \eqref{eq:term2} in \eqref{eq:3terms}, we get that for every $\e>0$,
\begin{equation}
\bigg\| \E{\Bigl(\sum_{i=1}^n{a_i^2{\bf Y}_i{\bf Y}_i^T}\Bigr)^{-1}} - \bigl(\E{ {\bf YY}^T}\bigr)^{-1} \bigg\|_{\mr{op}}  \lesssim_{{\bf Y}} \e+\frac{\log^\delta(d+1) \|a\|_{2+2\delta}^{2+2\delta}}{\e^{1+\delta}} + \frac{\log^{\frac{\delta^2}{1+\delta}}{(d+1)}\|a\|_{2+2\delta}^{2\delta}}{\e^\delta}
\end{equation}
where the implicit constant depends only on the moments of $\|{\bf YY}^T\|_{\mr{op}}$. Finally, the (almost) optimal choice $\e = \|a\|_{2+2\delta}^{\frac{2\delta}{1+\delta}}$ yields the desired bound.
\end{proof}

\begin{remark} \label{rem:norms}
We insisted on stating Theorem \ref{coropconvexity} as a bound for the operator norm of the (normalised) Fisher information matrix of $S_n$ but this is not necessary. An inspection of the proof reveals that given any norm $\|\cdot\|$ on $\M_d(\R)$ which is operator monotone, i.e.
\begin{equation}
0 \preceq A \preceq B \qquad \Longrightarrow \qquad \|A\| \leq \|B\|
\end{equation}
and satisfies the ideal property
\begin{equation}
\forall \ A,B\in \M_d(\R),\qquad \|AB\| \leq \|A\|_{\mr{op}} \|B\|,
\end{equation}
we can derive a bound of the form
\begin{equation} \label{eq:coropgen}
\big\|\mathrm{Cov}(S_n)^{\frac{1}{2}}\mathcal{I}(S_n)\mathrm{Cov}(S_n)^{\frac{1}{2}}-\mathrm{I}_d\big\| \leq C\big({\bf Y},\|\cdot\|\big)\ \|a\|_{2+2\delta}^{\frac{2\delta}{1+\delta}}
\end{equation}
for random matrices ${\bf Y}$ satisfying \eqref{matrixconditions}. The implicit constant depends on moments of $\|{\bf YY}^T\|$ and $\|{\bf YY}^T\|_{\mr{op}}$ and on the Rademacher type $(1+\delta)$-constant of $\|\cdot\|$. These conditions are, in particular, satisfied for all $\msf{S}_p^d$ norms and the corresponding type constant is subpolynomial in $d$ for $p\geq 1+\delta$.
\end{remark}

\begin{remark}
As was already mentioned in the introduction, bounding the relative Fisher information of a random vector automatically implies bounds for the relative entropy in view of the Gaussian logarithmic Sobolev inequality \cite{gross}. However, bounds for the Fisher information \emph{matrix} allow one to get better bounds for the relative entropy using more sophisticated functional inequalities which capture the whole spectrum of $\mc{I}(X)$. We refer to \cite{ES23} for more on this kind of inequalities.
\end{remark}

Finally, we present some examples of Gaussian mixtures related to conditions \eqref{matrixconditions}.

\medskip 

\noindent {\bf Examples.} {\bf 1.} Fix $p\in(0,2)$ and consider the random variable $X_p$ with density $c_pe^{-|x|^p}$, where $x\in\R$. It was shown in \cite[Lemma~23]{eskeGM} that $X$ can be expressed as 
\begin{equation}
X_p \stackrel{\mr{(d)}}{=} (2V_{\frac{p}{2}})^{-\frac{1}{2}}Z,
\end{equation}
where $V_{\frac{p}{2}}$ has density proportional to $t^{-\frac{1}{2}}g_{\frac{p}{2}}(t)$ and $g_a$ is the density of the standard positive $a$-stable law. The moments of $Y_p = V_{\frac{p}{2}}^{-1/2}$ then satisfy
\begin{equation}
\forall \alpha\in\R,\qquad \mb{E} Y_p^\alpha = \mb{E} V_{\frac{p}{2}}^{-\alpha/2} = \kappa_p \int_0^\infty t^{-\frac{\alpha+1}{2}} g_{\frac{p}{2}}(t) \,\diff t,
\end{equation}
for some $\kappa_p >0$. Since positive $\frac{p}{2}$-stable random variables have finite $\beta$-moments for all powers $\beta\in\big(-\infty,\frac{p}{2}\big)$, the assumptions \eqref{matrixconditions} are satisfied when 
\begin{equation}
\min\{2\delta+2,-2\delta-2\}>- p-1
\end{equation}
or, equivalently, $\delta<\frac{p-1}{2}$. Therefore, Theorem \ref{coropconvexity} applies for these variables when $p\in(1,2)$.

\smallskip

\noindent {\bf 2.} It is well-known (see, for instance \cite[Lemma~23]{eskeGM}) that for $p\in(0,2)$, the standard symmetric $p$-stable random variable $X_p$  can be written as
\begin{equation}
X_p \stackrel{\mr{(d)}}{=} (2G_{\frac{p}{2}})^{\frac{1}{2}}Z,
\end{equation}
where $G_{p/2}$ is a standard positive $\frac{p}{2}$ stable random variable. In this setting, the factor $G_{\frac{p}{2}}^{\frac{1}{2}}$ does not have a finite $2+2\delta$ moment for any value of $p$, so Theorem \ref{coropconvexity} does not apply.

\bibliographystyle{plain}
\bibliography{Info-GM.bib}

\end{document}